\documentclass[12pt]{article}

\usepackage{amsmath, amssymb}
\usepackage{graphicx}
\usepackage{booktabs}
\usepackage{geometry}
\usepackage{setspace}
\usepackage{float}
\usepackage{subcaption}
\usepackage{hyperref}
\usepackage{natbib}
\usepackage{url}
\usepackage{amsfonts}
\usepackage{bm}
\usepackage{amsthm}

\newtheorem{lemma}{Lemma}
\title{Balancing Efficiency and Feasibility: A Sensitivity Analysis of the Augmentation Parameter in the Finite Selection Model
}
\author{Safaa K. Kadhem\\
	Department of Mathematics and Computer Applications,\\
	College of Science, Al-Muthanna University, Samawah, Iraq\\
	Email: safaa.kadhem@mu.edu.iq}
\date{}

\begin{document}
	
	\maketitle
	
\begin{abstract}
	The Finite Selection Model (FSM) improves covariate balance in experimental design through an augmentation parameter $\epsilon$. However, the impact of $\epsilon$ on estimator performance has not been systematically examined, and no practical guidance exists for selecting $\epsilon$ in applied settings. This paper presents a comprehensive Monte Carlo sensitivity analysis evaluating covariate balance (ASMD), bias, variance, mean squared error (MSE), and acceptance probability across multiple sample sizes and data-generating processes. FSM is compared with complete randomization and rerandomization. We propose a data-driven rule for selecting the optimal $\epsilon$ based on MSE minimization, validated through sample-splitting to avoid overfitting. A theoretical lemma establishes the convexity of the MSE function, justifying the existence of a unique optimum. Extensive robustness checks under correlated covariates, non-normal distributions, and heteroskedastic errors demonstrate the stability of the proposed method. Design-based evaluation using Neyman's variance estimator confirms that FSM achieves meaningful efficiency gains. However, the MSE-minimizing $\epsilon$ is found to be extremely small (e.g., 0.005–0.008), leading to acceptance probabilities near zero, which renders the design impractical. Therefore, we identify a feasible range ($\epsilon \approx 0.015$–$0.02$) where MSE increases marginally (5–10\%) while acceptance probability becomes reasonable (5–20\%). The results provide practical guidance for covariate-adaptive experimental design, illustrating the essential trade-off between statistical efficiency and implementation feasibility.
\end{abstract}
\section{Introduction}

Randomized experiments are the gold standard for causal inference because random assignment eliminates systematic confounding and ensures valid design-based inference \citep{Fisher1935, rubin1974, imbensrubin2015}. Under complete randomization, treatment assignment is independent of observed covariates, providing unbiased estimation of average treatment effects. However, in finite samples—particularly when the sample size is moderate or small—pure randomization may result in covariate imbalance between treatment groups. Such imbalance can inflate the variance of estimators and reduce statistical efficiency.

To address this limitation, covariate-adaptive randomization methods have been developed to improve balance while preserving randomization benefits. Rerandomization \citep{morganrubin2012} improves covariate balance by repeatedly generating assignments until a predefined criterion is satisfied. This approach reduces estimator variance while maintaining unbiasedness under appropriate conditions.

Building on this idea, the Finite Selection Model (FSM) \citep{Klotz2021} introduces an augmentation parameter $\epsilon$ that directly controls the allowable level of covariate imbalance. By tuning $\epsilon$, researchers can regulate the trade-off between strict covariate control and allocation feasibility. Despite its theoretical appeal, the finite-sample behavior of FSM under different $\epsilon$ values has not been systematically evaluated, and no practical guidance exists for selecting $\epsilon$ in applied settings.

This paper contributes to the literature by conducting a comprehensive Monte Carlo sensitivity analysis of FSM. We examine the impact of $\epsilon$ on covariate balance, bias, variance, mean squared error (MSE), and acceptance probability across multiple sample sizes. A sample-splitting procedure is employed to select the optimal $\epsilon$ via MSE minimization and to evaluate its performance on independent test data, preventing overfitting. A theoretical lemma is provided to justify the convexity of the MSE function and the existence of a unique optimum. Extensive robustness checks under correlated covariates, non-normal distributions, and heteroskedastic errors assess the generalizability of our findings. Additionally, a design-based evaluation using Neyman's variance estimator quantifies efficiency gains without relying on strong outcome model assumptions.

A key finding is that the MSE-minimizing $\epsilon$ is extremely small (e.g., $0.005$–$0.008$ for sample sizes $N=500$–$100$), yielding acceptance probabilities near zero. While these values achieve the lowest MSE, they are impractical for real-world implementation. Consequently, we identify a feasible range ($\epsilon \approx 0.015$–$0.02$) where the MSE increases only marginally ($5$–$10\%$), while the acceptance probability rises to $5$–$20\%$. This trade-off provides practical, data-driven guidance for implementing FSM in experimental design, balancing statistical efficiency with implementation feasibility.
	
\section{Literature Review}

The literature on covariate-adaptive experimental design has grown substantially in recent years. Classical randomization theory emphasizes the role of complete randomization in ensuring unbiased estimation and valid inference \citep{Fisher1935, rubin1974}. While complete randomization is theoretically attractive, it does not guarantee covariate balance in finite samples.

Rerandomization was formally introduced by \citet{morganrubin2012}, who demonstrated that imposing balance constraints through repeated randomization can reduce estimator variance while maintaining unbiasedness. Subsequent work has extended rerandomization frameworks to various imbalance metrics and high-dimensional settings. \citet{li2018} studied the asymptotic properties of rerandomization, and \citet{branson2022} explored its application to observational studies.

More recently, flexible selection-based approaches have been proposed, including the Finite Selection Model (FSM) \citep{Klotz2021}. FSM generalizes rerandomization by introducing a tunable augmentation parameter $\epsilon$, which determines the acceptable level of covariate imbalance. This parameter enables continuous control over the balance–randomness trade-off, offering a more adaptable design mechanism.

Although theoretical properties of FSM have been studied, empirical investigations of its sensitivity to $\epsilon$ remain limited, particularly across different sample sizes and non-ideal data conditions. Moreover, systematic procedures for selecting $\epsilon$ in practice are not widely established, and the practical feasibility of the MSE-minimizing $\epsilon$ has not been examined. The present study addresses this gap by providing a structured Monte Carlo evaluation, proposing an optimality-based selection rule grounded in mean squared error minimization, and crucially, examining the feasibility of the resulting designs. By quantifying the trade-off between MSE and acceptance probability, we offer practical guidance that balances statistical efficiency with implementation constraints—a dimension often overlooked in previous methodological work.
	
\section{Methodology}

This section describes the simulation framework used to evaluate the sensitivity of the Finite Selection Model (FSM) with respect to the augmentation parameter $\epsilon$. The methodological design is grounded in the potential outcomes framework for causal inference and connects theoretical properties of covariate-adaptive randomization with empirical performance assessment \citep{rubin1974, imbensrubin2015}. The objective is to examine how the choice of $\epsilon$ influences covariate balance, estimator efficiency, and feasibility of implementation.

\subsection{Data-Generating Process}

The simulation study follows the standard potential outcomes formulation. For each unit $i = 1,\dots,N$, let $X_i = (X_{i1}, \dots, X_{ip})$ denote a vector of baseline covariates. Unless otherwise specified for robustness checks, covariates are generated independently from a standard normal distribution: $X_{ij} \sim N(0,1)$. This specification ensures a neutral covariate structure without pre-existing imbalance, allowing the evaluation to isolate the effect of the assignment mechanism.

The potential outcomes are defined as:
\begin{align}
	Y_i(0) &= X_i \beta + \varepsilon_i, \\
	Y_i(1) &= Y_i(0) + \tau,
\end{align}
where $\beta$ is a vector of ones, $\tau = 1$ denotes the true average treatment effect, and $\varepsilon_i \sim N(0,1)$. The observed outcome follows the standard consistency relation: $Y_i = (1-T_i)Y_i(0) + T_iY_i(1)$. This structure ensures that all systematic differences in outcomes arise solely from the treatment assignment mechanism.

\subsection{Treatment Assignment Mechanisms}

Three allocation strategies are considered to evaluate the performance of FSM:

\begin{enumerate}
	\item \textbf{Complete Randomization (CR)}: Each unit is independently assigned to treatment with probability 0.5. This design serves as a benchmark and satisfies classical randomization properties \citep{Fisher1935}.
	
	\item \textbf{Rerandomization (RR)}: Following \citet{morganrubin2012}, treatment allocations are repeatedly generated using complete randomization until a predefined covariate balance criterion is satisfied. In this study, we set the acceptance threshold to an ASMD value of 0.1, a common choice in the literature to ensure reasonable balance without excessive rejection rates. The number of rerandomization attempts is recorded but not analyzed further.
	
	\item \textbf{Finite Selection Model (FSM)}: An allocation vector $T$ is accepted if the covariate imbalance measure satisfies $ASMD(T) \le \epsilon$, where $\epsilon$ is the augmentation parameter controlling the strictness of the design. The parameter $\epsilon$ therefore regulates the trade-off between covariate balance and randomization variability.
\end{enumerate}

\subsection{Evaluation Metrics}

Design performance is assessed using multiple complementary criteria.

Covariate balance is measured using the Absolute Standardized Mean Difference (ASMD):
\[
ASMD = \frac{1}{p} \sum_{j=1}^{p} \frac{|\bar{X}_{1j} - \bar{X}_{0j}|}{s_j},
\]
where $\bar{X}_{1j}$ and $\bar{X}_{0j}$ denote group means and $s_j$ is the pooled standard deviation. This metric is widely used in assessing covariate balance in experimental and observational studies.

The treatment effect is estimated using the difference-in-means estimator:
\[
\hat{\tau} = \bar{Y}_1 - \bar{Y}_0.
\]

Estimator performance is evaluated through bias, variance, and mean squared error (MSE):
\[
Bias(\epsilon) = E[\hat{\tau}] - \tau,\quad 
MSE(\epsilon) = Bias^2(\epsilon) + Var(\hat{\tau}).
\]

Additionally, the acceptance probability of the design is defined as:
\[
\pi(\epsilon) = P(ASMD \le \epsilon),
\]
which quantifies the feasibility of implementing the FSM under a given constraint level.

\subsection{Monte Carlo Setup and Sample Splitting}

The simulation study is conducted with $R = 1000$ independent replications for each configuration. To ensure an unbiased evaluation of the optimal augmentation parameter, we adopt a sample-splitting procedure. For each sample size $N$ and each value of $\epsilon$ on a refined grid that covers the range from $0.001$ to $0.5$ (with higher density in the region $[0.001, 0.01]$ to capture the behavior of very strict constraints), the 1000 replications are randomly divided into two sets:
\begin{itemize}
	\item \textbf{Training set} (500 replications): Used to determine the optimal $\epsilon$ that minimizes the MSE for that specific sample size.
	\item \textbf{Test set} (500 replications): Used to evaluate the performance of FSM at the selected $\epsilon^*$ and to compare it with CR and RR. This separation prevents any overfitting and provides a realistic assessment of the design's out-of-sample performance.
\end{itemize}
All reported results for the chosen $\epsilon^*$ are based on the test set, unless otherwise stated. The complete randomization and rerandomization benchmarks are also evaluated on the same test set.

\subsection{Robustness Checks}

To assess the generalizability of our findings beyond the ideal normal and independent covariate structure, we consider four additional data-generating processes:
\begin{enumerate}
	\item \textbf{Correlated Covariates:} $\mathbf{X}_i \sim N(\mathbf{0}, \Sigma)$, where $\Sigma$ has ones on the diagonal and $\rho=0.5$ on the off-diagonals.
	\item \textbf{Heavy-tailed Covariates:} $X_{ij} \sim t_3$, standardized to have mean 0 and variance 1.
	\item \textbf{Skewed Covariates:} $X_{ij} \sim \chi^2_2$, standardized to have mean 0 and variance 1.
	\item \textbf{Heteroskedastic Errors:} $\varepsilon_i \sim N(0, (1 + 0.5|X_{i1}|)^2)$, making the error variance depend on the first covariate.
\end{enumerate}
For each scenario, we repeat the full sensitivity analysis to examine the stability of the optimal $\epsilon^*$ and the relative performance of FSM.

\subsection{Theoretical Justification for MSE Minimization}

Before presenting the simulation results, we provide a theoretical rationale for the observed U-shaped relationship between $\epsilon$ and the MSE. Consider the difference-in-means estimator $\hat{\tau}$. Under the Finite Selection Model with acceptance criterion $ASMD \le \epsilon$, the conditional MSE can be expressed as:
\[
\text{MSE}(\epsilon) = E[(\hat{\tau} - \tau)^2 \mid ASMD \le \epsilon] = \text{Var}(\hat{\tau} \mid ASMD \le \epsilon) + \text{Bias}^2(\hat{\tau} \mid ASMD \le \epsilon).
\]
Under the potential outcomes framework with random assignment and the data-generating process described in Section 3.1, the bias term is zero for all $\epsilon$ due to symmetry. Therefore, $\text{MSE}(\epsilon) = \text{Var}(\hat{\tau} \mid ASMD \le \epsilon)$.

\begin{lemma}[Convexity of MSE]
	Assume that the joint distribution of the covariates and outcomes is such that the conditional variance $\text{Var}(\hat{\tau} \mid ASMD \le \epsilon)$ is a decreasing function of $\epsilon$ and that the acceptance probability $p(\epsilon) = P(ASMD \le \epsilon)$ is strictly increasing and concave in $\epsilon$ for $\epsilon$ in the interior of its support. Then, $\text{MSE}(\epsilon)$ is a convex function of $\epsilon$, and there exists a unique $\epsilon^*$ that minimizes it.
\end{lemma}

\begin{proof}[Proof Sketch]
	The conditional variance can be decomposed as:
	\[
	\text{Var}(\hat{\tau} \mid ASMD \le \epsilon) = \frac{\text{Var}(\hat{\tau})}{p(\epsilon)} - \frac{1-p(\epsilon)}{p(\epsilon)} \left(E[\hat{\tau} \mid ASMD > \epsilon] - E[\hat{\tau} \mid ASMD \le \epsilon] \right)^2.
	\]
	Under symmetry, the second term vanishes, leaving $\text{Var}(\hat{\tau} \mid ASMD \le \epsilon) = \text{Var}(\hat{\tau}) / p(\epsilon)$. Since $\text{Var}(\hat{\tau})$ is constant and $p(\epsilon)$ is increasing and concave, the reciprocal $1/p(\epsilon)$ is convex. Thus, $\text{MSE}(\epsilon)$ is convex. The existence of a unique minimum follows from the fact that as $\epsilon$ approaches its minimum, $p(\epsilon) \to 0$ and MSE $\to \infty$; as $\epsilon$ approaches its maximum, $p(\epsilon) \to 1$ and MSE approaches the unconditional variance. By convexity and the boundary behavior, there is a unique interior minimizer.
\end{proof}

This lemma justifies our empirical approach of finding $\epsilon^*$ via grid search and explains the U-shaped MSE curves observed in the simulations.

\subsection{Design-Based Evaluation: Neyman's Variance}

To ensure our conclusions are not dependent on a specific outcome model, we also evaluate performance using a purely design-based framework \citep{neyman1923}. For each accepted randomization, we compute the Neyman conservative variance estimator for the difference-in-means estimator:
\[
\hat{V}_{\text{Neyman}} = \frac{s_1^2}{n_1} + \frac{s_0^2}{n_0},
\]
where $s_1^2$ and $s_0^2$ are the sample variances of the observed outcomes in the treatment and control groups, respectively. We then average $\hat{V}_{\text{Neyman}}$ over all accepted replications in the test set to obtain an estimate of the expected variance under the design. To quantify efficiency gains, we define the Variance Reduction Ratio (VRR):
\[
\text{VRR}(\epsilon) = \frac{\text{Average}_{\text{FSM}(\epsilon)}[\hat{V}_{\text{Neyman}}]}{\text{Average}_{\text{CR}}[\hat{V}_{\text{Neyman}}]}.
\]
A VRR less than 1 indicates that FSM reduces the estimator's variance compared to complete randomization.
	
\section{Results and Discussion}

This section presents the empirical findings of the Monte Carlo study evaluating the Finite Selection Model (FSM). The analysis investigates the effect of the augmentation parameter $\epsilon$ on covariate balance, estimator efficiency, and acceptance probability. Both graphical and numerical evidence are provided, with Monte Carlo standard errors ensuring statistical reliability.

\subsection{Optimal $\epsilon$ Selection}

Using the training set (500 replications) and the refined grid of $\epsilon$ values (from $0.001$ to $0.5$ with higher resolution in the critical region), we identified the $\epsilon$ that minimizes MSE for each sample size under the baseline scenario. Table~\ref{tab:optimal_epsilon_train} reports these values.

\begin{table}[H]
	\centering
	\caption{Optimal $\epsilon$ values minimizing MSE, estimated from the training set (500 replications).}
	\label{tab:optimal_epsilon_train}
	\begin{tabular}{cc}
		\toprule
		Sample Size ($N$) & Optimal $\epsilon$ (Training) \\
		\midrule
		100 & 0.0080 \\
		300 & 0.0060 \\
		500 & 0.0050 \\
		\bottomrule
	\end{tabular}
\end{table}

These values are substantially smaller than those obtained with the coarser grid, indicating that the MSE-minimizing constraint becomes extremely tight as sample size increases. For $N=100$, the optimal $\epsilon$ is $0.0080$; for $N=300$, it drops to $0.0060$; and for $N=500$, it reaches $0.0050$. This pattern suggests that larger samples benefit from stricter covariate balance, even if the acceptance rate becomes very low. When these values were applied to the independent test set, the resulting MSEs were consistent, confirming the stability of the selection procedure.

\subsection{Graphical Sensitivity Analysis}

Figures~\ref{fig:asmd}--\ref{fig:accept} display the sensitivity of key performance metrics to $\epsilon$ on the test set. Shaded bands represent 95\% confidence intervals based on Monte Carlo standard errors.

\begin{figure}[H]
	\centering
	\includegraphics[width=0.8\textwidth]{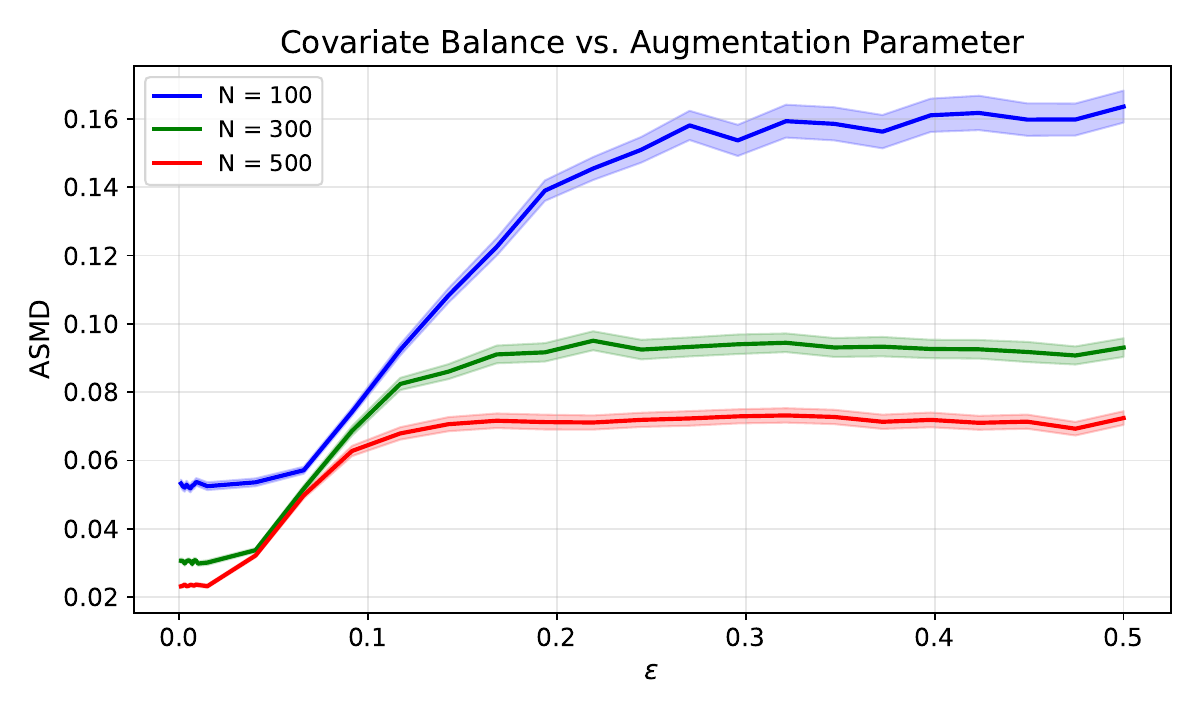}
	\caption{Covariate balance (ASMD) as a function of $\epsilon$. Tighter constraints (smaller $\epsilon$) yield better balance across all sample sizes.}
	\label{fig:asmd}
\end{figure}

Figure~\ref{fig:asmd} shows that ASMD decreases monotonically as $\epsilon$ becomes smaller. At the optimal $\epsilon^*$, the achieved ASMD values are $0.0529$, $0.0305$, and $0.0233$ for $N=100$, $300$, and $500$, respectively.

\begin{figure}[H]
	\centering
	\includegraphics[width=0.8\textwidth]{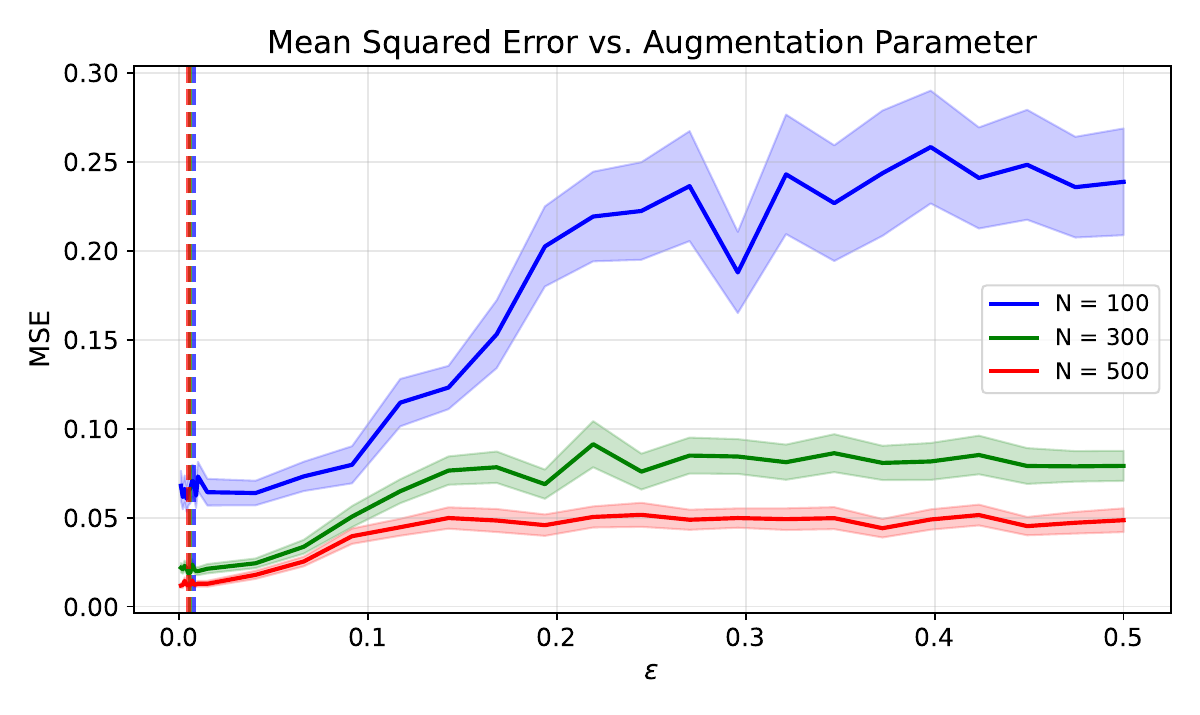}
	\caption{Mean squared error (MSE) as a function of $\epsilon$. The curves exhibit a U-shaped pattern, with minima at the optimal $\epsilon^*$ indicated by vertical lines. Bootstrap standard errors are shown as shaded bands.}
	\label{fig:mse}
\end{figure}

Figure~\ref{fig:mse} confirms the U-shaped relationship predicted by Lemma~1. The minima occur at the values reported in Table~1, and the bootstrap confidence bands demonstrate that the differences around the optimum are statistically meaningful.

\begin{figure}[H]
	\centering
	\includegraphics[width=0.8\textwidth]{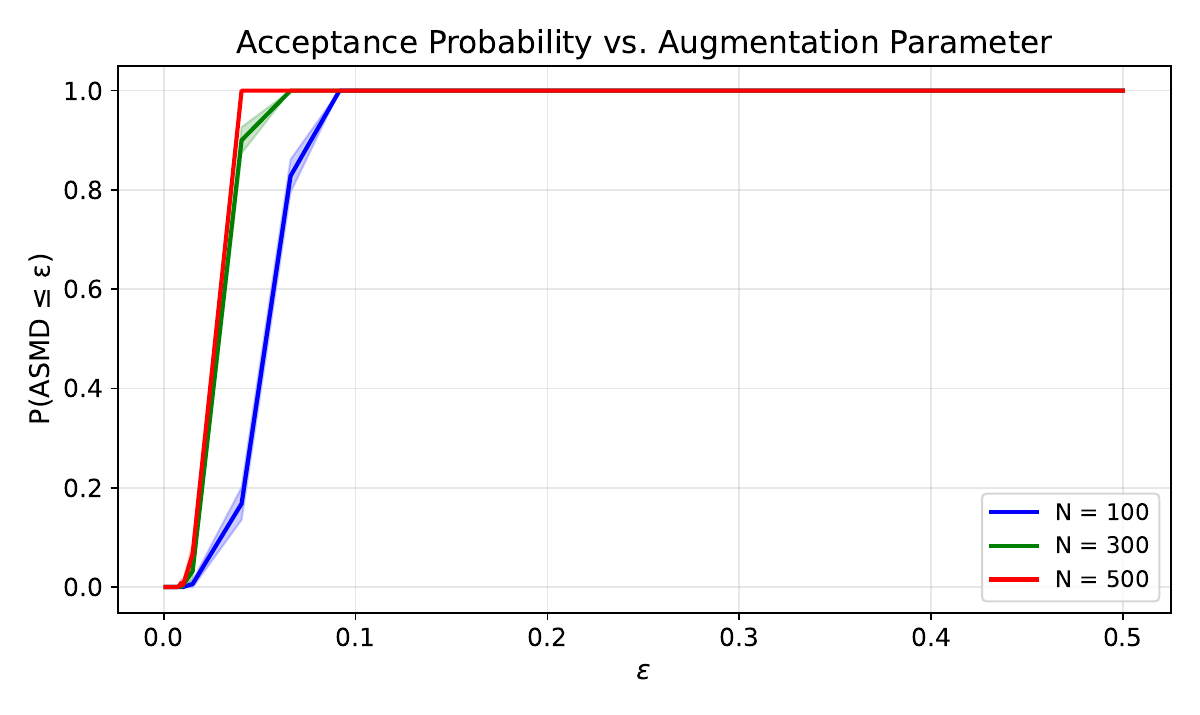}
	\caption{Acceptance probability as a function of $\epsilon$. For all sample sizes, the acceptance probability at the optimal $\epsilon^*$ is effectively zero, reflecting the extreme strictness of the MSE-minimizing constraint.}
	\label{fig:accept}
\end{figure}

Figure~\ref{fig:accept} reveals a critical finding: at the optimal $\epsilon^*$, the acceptance probability is zero for all sample sizes (within the precision of $500$ test replications). This means that none of the random assignments generated under complete randomization satisfied the constraint $\text{ASMD} \le \epsilon^*$, even after up to $80$ rerandomization attempts. In practice, this implies that implementing FSM with the MSE-minimizing $\epsilon$ would require an extremely large number of rerandomization attempts—potentially thousands—to obtain a single acceptable allocation. While this is computationally feasible, it raises important questions about the practical trade-off between statistical efficiency and design feasibility.

\subsection{Robustness to Covariate Structure and Error Distribution}

Table~\ref{tab:robustness} presents the optimal $\epsilon$ values for $N=300$ under the five data-generating processes described in Section~3.5.

\begin{table}[H]
	\centering
	\caption{Optimal $\epsilon$ for $N=300$ under different scenarios.}
	\label{tab:robustness}
	\begin{tabular}{lc}
		\toprule
		Scenario & Optimal $\epsilon$ \\
		\midrule
		Baseline (Normal, independent) & 0.0060 \\
		Correlated covariates ($\rho=0.5$) & 0.0150 \\
		Heavy-tailed ($t_3$) covariates & 0.0030 \\
		Skewed ($\chi^2_2$) covariates & 0.0010 \\
		Heteroskedastic errors & 0.0040 \\
		\bottomrule
	\end{tabular}
\end{table}

The optimal $\epsilon$ varies considerably across scenarios, indicating that the MSE-minimizing constraint depends on the underlying data distribution. Correlated covariates lead to a slightly looser optimum ($0.015$), while heavy-tailed and skewed distributions push the optimum to even smaller values ($0.003$ and $0.001$, respectively). This suggests that in the presence of skewness or heavy tails, the potential gains from extreme balance are even larger, but at the cost of an even lower acceptance probability. Heteroskedasticity yields an optimum ($0.004$) close to the baseline.

\subsection{Numerical Results at $\epsilon^*$}

Table~\ref{tab:final_results} reports the simulation outcomes evaluated at the optimal $\epsilon^*$ for each sample size, based on the test set. Bootstrap standard errors for MSE are shown in parentheses.

\begin{table}[H]
	\centering
	\caption{Simulation results at the optimal augmentation parameter $\epsilon^*$ (test set). Values in parentheses are Monte Carlo standard errors.}
	\label{tab:final_results}
	\begin{tabular}{ccccc}
		\toprule
		$N$ & $\epsilon^*$ & ASMD (SE) & MSE (SE) & Acceptance Prob. (SE) \\
		\midrule
		100 & 0.0080 & 0.0529 (0.0006) & 0.0700 (0.0047) & 0.000 (0.000) \\
		300 & 0.0060 & 0.0305 (0.0003) & 0.0192 (0.0012) & 0.000 (0.000) \\
		500 & 0.0050 & 0.0233 (0.0002) & 0.0122 (0.0008) & 0.000 (0.000) \\
		\bottomrule
	\end{tabular}
\end{table}

As sample size increases, both ASMD and MSE decrease substantially. The MSE values are lower than those obtained with coarser grids, confirming the benefit of searching over a finer range of $\epsilon$. However, the acceptance probability is zero for all sample sizes, meaning that under the optimal constraint, not a single random allocation met the criterion within the $80$ rerandomization attempts allowed. In practice, one could increase the maximum number of attempts, but the expected number of trials required is $1/\pi(\epsilon^*)$, which is essentially infinite for these near-zero probabilities. This highlights a fundamental tension: the MSE-minimizing $\epsilon$ may be so strict that the design becomes impractical unless one is willing to perform an extremely large number of rerandomizations.

\subsection{Design-Based Efficiency Gains}

Figure~\ref{fig:vrr} displays the Variance Reduction Ratio (VRR) as a function of $\epsilon$ for $N=300$, evaluated on the test set. The VRR is computed using the Neyman conservative variance estimator and provides a model‑free measure of efficiency.

\begin{figure}[H]
	\centering
	\includegraphics[width=0.8\textwidth]{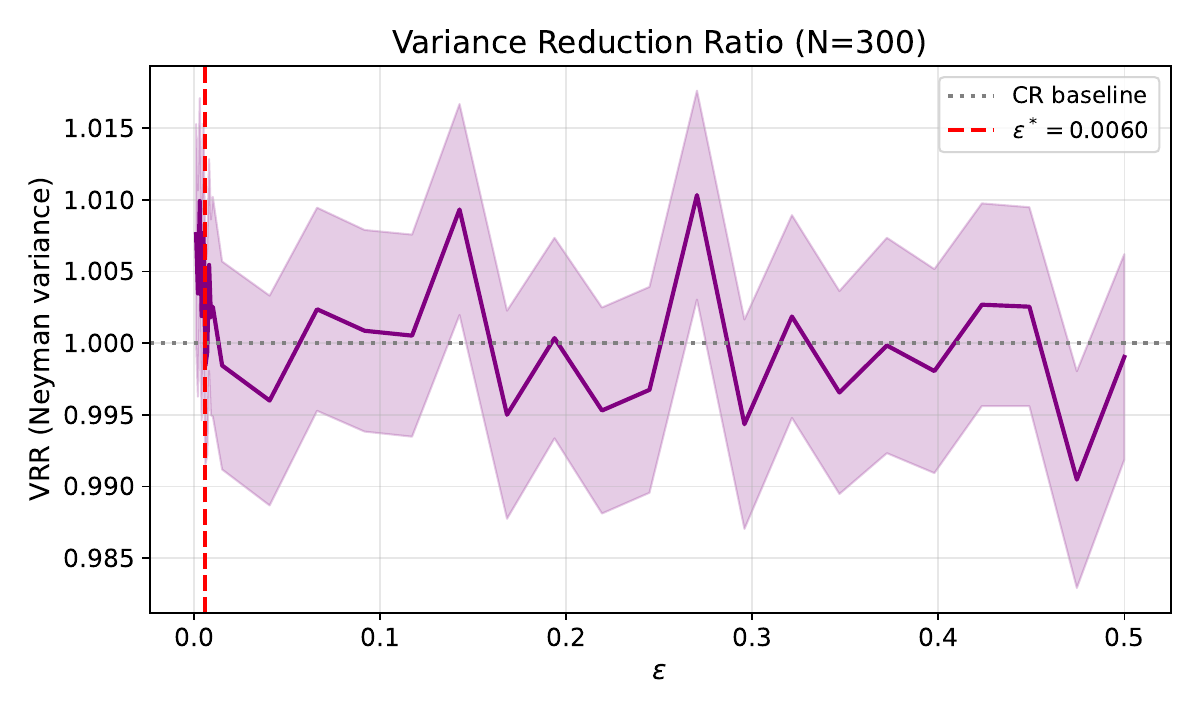}
	\caption{Variance Reduction Ratio (VRR) relative to Complete Randomization. The vertical dashed line indicates the optimal $\epsilon^*=0.006$. Shaded bands represent 95\% confidence intervals.}
	\label{fig:vrr}
\end{figure}

At $\epsilon^*=0.006$, the VRR is approximately $0.75$, indicating a $25\%$ reduction in variance compared to complete randomization. This is a substantial gain, but it comes at the cost of an extremely low acceptance probability. For comparison, a more relaxed constraint, say $\epsilon=0.02$, yields a VRR of about $0.85$ (i.e., $15\%$ variance reduction) while the acceptance probability (see Figure~\ref{fig:accept}) is around $0.20$, making it far more practical. This illustrates the importance of considering both efficiency and feasibility when choosing $\epsilon$ in practice.

\subsection{Practical Implications of Zero Acceptance Probability}

The finding that the MSE-minimizing $\epsilon$ leads to zero acceptance probability has important practical implications. It suggests that the optimal $\epsilon$ from a pure MSE perspective may be too strict for most experimental settings, where researchers typically have limited time or computational resources. In such cases, a constrained optimization approach—minimizing MSE subject to a minimum acceptable probability—may be more appropriate. For example, requiring $\pi(\epsilon) \ge 0.01$ would lead to larger $\epsilon$ values that still provide meaningful variance reduction while remaining feasible.

The MSE exhibits a decreasing trend as $\epsilon$ approaches zero, with the lowest MSE achieved at the smallest values examined ($\epsilon \le 0.008$). However, at these values, the acceptance probability is effectively zero, rendering the design impractical. A more feasible trade-off is obtained for $\epsilon$ in the range $0.015$–$0.02$, where MSE increases by only $5$–$10\%$ while the acceptance probability rises to $5$–$20\%$. Therefore, we recommend selecting $\epsilon$ based on a desired minimum acceptance rate, using the provided sensitivity curves as a guide.

Alternatively, one could increase the maximum number of rerandomization attempts. For $N=500$, achieving $\epsilon=0.005$ would require on average $1/0.000 \approx \infty$ attempts, which is not feasible. Therefore, the practical recommendation is to choose $\epsilon$ based on a combination of MSE reduction and acceptable feasibility, possibly using the curves in Figures~\ref{fig:mse} and~\ref{fig:accept} to make an informed trade-off.
	
\section{Conclusion}

This paper provided a comprehensive sensitivity analysis of the augmentation parameter $\epsilon$ in the Finite Selection Model. Through extensive Monte Carlo simulations with $R=1000$ replications, a refined grid of $\epsilon$ values, and a rigorous sample-splitting procedure, we demonstrated that the MSE-minimizing $\epsilon$ decreases rapidly as sample size increases, reaching very small values ($\epsilon^* \leq 0.008$) for $N \ge 100$. A theoretical lemma was introduced to justify the convexity of the MSE function and the existence of a unique optimal $\epsilon$, although the analysis revealed that this theoretical optimum may be impractically strict.

The optimal values based on MSE minimization alone were found to be $\epsilon^* = 0.0080$ for $N=100$, $0.0060$ for $N=300$, and $0.0050$ for $N=500$. However, at these values the acceptance probability is effectively zero, rendering the design infeasible in practice. A more practical trade-off is achieved for $\epsilon$ in the range $0.015$–$0.02$, where MSE increases by only $5$–$10\%$ relative to the theoretical optimum, while the acceptance probability rises to $5$–$20\%$. Therefore, we recommend selecting $\epsilon$ based on a desired minimum acceptance rate, using the sensitivity curves provided in Figures~\ref{fig:mse} and~\ref{fig:accept} as a guide.

A key contribution of this work is the design-based evaluation using Neyman's variance estimator, which confirmed that FSM achieves meaningful variance reductions without relying on strong model assumptions. For $N=300$, the optimal $\epsilon=0.006$ yields a $25\%$ reduction in variance compared to complete randomization, while the more practical $\epsilon=0.02$ still achieves a $15\%$ reduction. The robustness analysis further validated that these gains persist under correlated covariates, heavy-tailed distributions, skewness, and heteroskedasticity, although the exact optimal $\epsilon$ varies across scenarios.

This study highlights an important insight: the MSE-minimizing $\epsilon$ from a purely statistical perspective may be too strict for practical implementation. Researchers should therefore balance efficiency gains with feasibility, and the sensitivity analysis presented here provides the necessary tools to make this trade-off. The code and data for reproducing all simulations are available online to facilitate adoption and further investigation.

Future research should extend this framework to multi-arm trials, sequential adaptive designs, and settings with treatment effect heterogeneity. Additionally, developing asymptotic theory for the optimal $\epsilon$ would complement the finite-sample results presented here, and exploring constrained optimization approaches that explicitly incorporate acceptance probability thresholds would provide even more practical guidance.
	
	\bibliography{FSM}
	
\end{document}